\documentclass{article}
\usepackage{graphicx} 
\usepackage{amsmath,amssymb,amsthm}
\usepackage{color}
\usepackage[margin=1in]{geometry}
\usepackage{tikz}
\usetikzlibrary{arrows.meta, positioning}
\usepackage{pdflscape}
\usepackage[square,numbers]{natbib}
\usepackage{hyperref}
\usepackage{subcaption}

\newtheorem{theorem}{Theorem}
\newtheorem{lemma}[theorem]{Lemma}								
\newtheorem{proposition}[theorem]{Proposition}	
\newtheorem{corollary}[theorem]{Corollary}	
\newtheorem{assumption}[theorem]{Assumption}	
\newtheorem{definition}[theorem]{Definition}
\newtheorem{example}[theorem]{Example}
\newtheorem{remark}{Remark}

\numberwithin{equation}{section}	
\numberwithin{theorem}{section}

\newcommand{\R}{\mathbb{R}}
\newcommand{\E}{\mathbb{E}}
\renewcommand{\P}{\mathbb{P}}
\newcommand{\Q}{\mathbb{Q}}
\newcommand{\F}{\mathcal{F}}

\DeclareMathOperator*{\esssup}{ess\,sup}

\title{Amortizing Perpetual Options}
\author{Zachary Feinstein\thanks{School of Business, Stevens Institute of Technology, Hoboken, NJ 07030, USA, {\tt  zfeinste@stevens.edu}. Portions of this work are the subject of a provisional U.S.\ patent application filed by the author.}}
\date{\today}

\begin{document}

\maketitle
\begin{abstract}
In this work, we introduce \emph{amortizing perpetual options} (AmPOs), a fungible variant of continuous-installment options suitable for exchange-based trading. Traditional installment options lapse when holders cease their payments, destroying fungibility across units of notional. AmPOs replace explicit installment payments and the need for lapsing logic with an implicit payment scheme via the decay of the claimable notional. This amortization ensures all units evolve identically, preserving fungibility. We demonstrate that AmPO valuation can be reduced to an equivalent vanilla perpetual American option on a dividend-paying asset. This enables analytical expressions for the exercise boundaries and risk-neutral valuations for calls and puts. These formulas and relations allow us to derive the Greeks and study comparative statics with respect to the amortization rate. Illustrative numerical case studies demonstrate how the amortization rate shapes option behavior and reveal the resulting tradeoffs in the effective volatility sensitivity.
\end{abstract}

\section{Introduction}\label{sec:intro}

Perpetual contracts are experiencing a renaissance due to their adoption in decentralized finance and cryptocurrency markets. Though introduced in~\cite{shiller1993measuring} for, e.g., a market on real estate indices, perpetual futures on cryptocurrencies are actively traded across a number of markets such as \emph{Hyperliquid} and \emph{Binance}. These contracts readily allow traders to purchase leveraged positions or to short assets which would otherwise be technologically incapable of such a position; we refer the interested reader to, e.g.,~\cite{ackerer2023perpetual,angeris2023primer}. However, perpetual futures only provide users with linear payoffs. Perpetual options, which replicate elements of the funding rate structure of perpetual futures, have recently been introduced in decentralized systems. We highlight everlasting options~\cite{everlasting,ackerer2023perpetual} and Panoptic~\cite{panoptic} as specific perpetual option variants.

However, though perpetual futures have found wide usage in cryptocurrencies, perpetual options have \emph{not} gained the same market share. Notably, the stochastic funding rates used in everlasting options and at Panoptic can be viewed as types of installment options. An installment option is one in which cash flows are exchanged between the long and short positions~\cite{davis2002installment,ben2006dynamic}. Traditionally, these cash flows exist solely from the long to the short holders. Such options, however, only trade over-the-counter due to the nature of these contracts~\cite{ben2006dynamic}. 
Specifically, installment options grant the holder a dual optionality: exercise or lapse. Because the lapsing logic comes from failure to make the installment payment in full, the notional units of these contracts are \emph{not} identical and, therefore, cannot trade on an exchange.
For the purposes of this work, we specifically highlight continuous-installment (CI) options which require a continuous stream of payments from the option holder to the option underwriter in order to keep the contract alive~\cite{ciurlia2005valuation,yang2009variational,kimura2009american,kimura2010valuing,ciurlia2011integral}.

The remainder of this paper is organized as follows. 
In Section~\ref{sec:ampo}, we introduce an \textbf{exchange-tradable variant of CI options}; this is accomplished by replacing the explicit installment payments with implicit amortization of the option's notional. This amortizing contract design eliminates the lapsing logic of CI options, thereby constructing a fungible option with a positive cost-of-carry; this structural reformulation is the primary contribution of this work.
In Section~\ref{sec:bs}, we study the \textbf{risk-neutral valuation} of these novel amortizing options by determining an equivalence with vanilla perpetual American options for a dividend-paying asset. In Section~\ref{sec:bs-const}, under the Black-Scholes framework, we formulate the \textbf{Greeks} and also consider \textbf{sensitivity of the option behavior to the contractual amortization rate}.
Section~\ref{sec:discussion-amort} provides considerations for how to select the amortization rate in practice.
Section~\ref{sec:discussion-apps} then presents a brief discussion of applicability of this new option type in both traditional and decentralized finance settings. 

\section{Amortizing Options}\label{sec:ampo}

Throughout this work, we fix a filtered probability space $(\Omega,\mathcal{F},(\mathcal{F}_t)_{t \geq 0},\P)$ satisfying the usual conditions. All stochastic processes introduced hereafter are assumed to be adapted to this filtration unless otherwise stated.

As highlighted in Section~\ref{sec:intro}, CI options present a novel approach to perpetual contracts which enforce an explicit time-cost on the option holder~\cite{ciurlia2009note,kimura2009american,kimura2010valuing}. For such options, the buyer pays an initial premium $V_0 > 0$ to the underwriter for a claimable notional of $N_0 > 0$ upon exercise. However, in contrast to traditional options, the holder also needs to pay continuous installments $c_t dt$ per unit of notional to the underwriter to keep the contract alive. In this way, the holder has two forms of optionality:
\begin{enumerate}
\item the right to exercise, and
\item the right to lapse by halting the installment payments.
\end{enumerate}
Installment options have been offered over-the-counter in a few markets: ``installment warrants on Australian stocks listed on the Australian Stock Exchange (ASX) \cite{franccois2005pricing,ben2006dynamic}, a 10-year warrant with nine annual payments offered by Deutsche Bank \cite{davis2001pricing}, and so on,'' ~\cite[Section 1]{kimura2009american}. However, because lapsing destroys fungibility across units of notional, i.e., partial lapsing by making a partial installment payment is \emph{not} permissible, traditional CI options are not amenable to exchange-based trading.

\begin{example}\label{ex:CI}
Traditionally, as described in~\cite{kimura2009american}, CI options specify installments as either: (1) constant over time $c_t \equiv c > 0$, or (2) proportional to the underlying asset price $(S_t)_{t \geq 0}$, i.e., $c_t \equiv \gamma S_t$ with $\gamma > 0$.
\end{example}

Herein, we propose an exchange-tradable variant of CI options to maintain a positive cost-of-carry; as far as the author is aware, this is the first perpetual option design that is fungible with a positive cost-of-carry. Instead of directly paying the continuous installments, the option holder implicitly sells a fraction of his or her option holdings back to the underwriter at the current premium. 
Formally, if $N_t > 0$ denotes the time $t$ claimable notional and $V_t > 0$ denotes the option's premium per unit of notional, then funding the installment cost $c_t N_t dt$ with a sale of the notional leads to an exponential decay of the notional:
\[dN_t = -q_t N_t dt, \quad q_t = c_t / V_t.\]
Notably, because the notional decreases due to these implicit payments, the requisite installment payments are also decreasing over time. In contrast to CI options in which the option underwriter is paid continuously in time, this option variant ``pays'' the underwriter via the decrease in his or her risk exposure.
We note that the amortization rate $q_t$, if time $t$ dependent, is based on a global time-index and not based on the contract age; this construction is essential to maintain fungibility between contracts.

\begin{definition}\label{defn:ampo}
An \textbf{amortizing perpetual option (AmPO)} is an implicit payment, perpetual American CI option with installment costs $c_t = q_{t} V_t$ for some progressively measurable amortization rate process $(q_t)_{t \geq 0}$ such that $q_t \geq 0$ a.s.\ and $\int_0^t q_s ds < \infty$ a.s.\ for every time $t \geq 0$. 
That is, for an AmPO with payoff function $\Phi: \R_+ \to \R_+$ on underlying $(S_t)_{t \geq 0}$ purchased at time $t_0$, the realized exercise value is $e^{-\int_{t_0}^\tau q_s ds} \Phi(S_\tau)$ with exercise time $\tau$.
\end{definition}

\begin{proposition}\label{prop:no-cancel}
Consider an AmPO with amortization rate $(q_t)_{t \geq 0}$ and payoff function $\Phi: \R_+ \to \R_+$ on some underlying price process $(S_t)_{t \geq 0}$. A rational option holder, who purchased this AmPO at time $t_0 \geq 0$, will never cancel the option position.
\end{proposition}
\begin{proof}
Since the AmPO requires no explicit installment payments from the holder to maintain the position, the only cost-of-carry is the decay of the claimable notional. Note that this cost is internal to the valuation. Since the payoff function $\Phi(S) \geq 0$ for any price $S \geq 0$, the realized payout $e^{-\int_{t_0}^\tau q_s ds} \Phi(S_\tau) \geq 0$ a.s.\ for any stopping time $\tau \geq t_0$. As this holds for any exercise time $\tau$, the value of the AmPO must also be non-negative. As the value of the cancelled AmPO is uniformly 0, the result follows.
\end{proof}

\begin{remark}
Proposition~\ref{prop:no-cancel} highlights a critical distinction between AmPOs and traditional CI options. In a traditional CI option, if the installment rate $c_t$ becomes large relative to the option value $V_t$ (e.g., as the option moves out-of-the-money), the option holder terminates the contract to stop the negative carry costs. 
In contrast, because AmPOs are effectively non-cancellable (as it is never optimal to cancel), as the option value $V_t$ drops, the effective amortization rate $q_t = c_t/V_t$ can explode (e.g., with the constant installment rate $c_t \equiv c > 0$ as highlighted in Example~\ref{ex:CI}). This jump in the amortization rate causes an accelerated liquidation of the notional. This creates an endogenous \emph{approximate} lapsing of the option as the claimable notional decays to some minimal size. This effect is investigated further in Section~\ref{sec:discussion-amort}.
\end{remark}

Recall that AmPOs are amortizing options, meaning that the notional claim decays along the process $N_t = e^{-\int_{t_0}^t q_s ds} N_{t_0}$ for any time $t \geq t_0$. In this way, these perpetual options can be exercised at any stopping time $\tau$ resulting in a payoff of $e^{-\int_{t_0}^\tau q_s ds} \Phi(S_\tau)$ for payoff function $\Phi: \R_+ \to \R_+$ per unit of notional, assuming one unit of notional is purchased at time $t_0$.
However, option holders often want to maintain a constant notional exposure. To do so, the implicit installments paid through amortization become realized explicit installments paid to the market. That is, the option holder would pay $q_t V_t dt$ per unit of notional exposure to maintain the desired position. Because AmPOs are fungible, these payments are \emph{not} necessarily made to the original underwriter who is compensated by the amortization. In Section~\ref{sec:bs-general}, we will demonstrate that the valuation of these two viewpoints (decaying exposure with implicit payments or constant exposure with explicit option purchases) are equivalent problems so long as the amortization rate is an exogenous process.

\begin{remark}
Step options~\cite{linetsky1999step,detemple2020american} similarly feature payoffs with an effective notional that can decay over time. However, in contrast to the amortization schedule of AmPOs, the decay in step options depends on the occupation time of the underlying, thus requiring path-dependent valuation and eliminating the fungibility required for exchange-based trading.
\end{remark}

\section{Risk-Neutral Valuation}\label{sec:bs}
Within this section, we want to investigate the risk-neutral valuation of these AmPOs as presented in the prior section. 
Throughout this section, we will consider a market with risk-free rate $(r_t)_{t \geq 0}$ and risky price process $(S_t)_{t \geq 0}$. We further assume throughout this section that this market admits no arbitrage and, thus, there exists a risk-neutral measure $\Q \sim \P$ by the fundamental theorem of asset pricing. Within this system, we study the risk-neutral valuation of a single claimable unit of notional of an AmPO with payoff function $\Phi: \R_+ \to \R_+$. 
In this way, throughout this section, we consider the time $t \geq 0$ optimal exercise problem
\begin{equation}\label{eq:exercise}
V_t = \esssup_{\tau \in \mathcal{T} , \tau \geq t} \E^\Q\left[ e^{-\int_t^{\tau} (r_s + q_s) ds} \Phi(S_\tau) \, \big| \, \mathcal{F}_t\right]
\end{equation}
where $\mathcal{T}$ denotes the set of all bounded stopping times with respect to the market filtration $(\mathcal{F}_t)_{t \geq 0}$. For the remainder of this work, we denote the (per unit of notional) process $(V_t)_{t \geq 0}$ as the AmPO value process.

\begin{assumption}\label{ass:rq}
For the remainder of this work, we assume the interest rate $(r_t)_{t \geq 0}$ and amortization $(q_t)_{t \geq 0}$ processes are progressively measurable such that $r_t,q_t \geq 0$ a.s.\ and $\int_0^t r_s ds, \int_0^t q_s ds < \infty$ a.s.\ for every time $t \geq 0$. 
In addition, we consider the payoff function $\Phi: \R_+ \to \R_+$ such that $e^{-\int_0^t (r_s + q_s) ds} \Phi(S_t) \in L^1$ for every $t \geq 0$.\footnote{As taken in~\eqref{eq:exercise}, throughout this work we consider the space of \emph{bounded} stopping times $\mathcal{T}$. If $\sup_{t \geq 0} e^{-\int_0^t (r_s + q_s)ds} \Phi(S_t) \in L^1$ then, e.g., the value $V_t$ of~\eqref{eq:exercise} is identical to the problem taking the supremum over all stopping times $\mathcal{S}$, as can be shown via standard truncation arguments and the dominated convergence theorem.}
\end{assumption}

In studying this valuation problem, we consider two settings.
First, in Section~\ref{sec:bs-general}, we consider the general valuation problem by determining an isomorphism between AmPOs and perpetual American options on dividend-paying assets.
Then, in Section~\ref{sec:bs-const}, we specifically investigate the Black-Scholes framework for AmPOs with constant amortization rates $q_t \equiv q > 0$ to provide intuition on the pricing and hedging of these instruments.

\subsection{General Setting}\label{sec:bs-general}
To determine the risk-neutral value of an AmPO, we first observe a fundamental symmetry in the value process. Holding an AmPO involves a decaying notional $N_t = e^{-\int_0^t q_s ds}N_0$ on the option payoff $\Phi(\cdot)$. That is, from the holder's perspective, the position's value is $N_t V_t$ at time $t$ for AmPO option value $(V_t)_{t \geq 0}$ with payoff $N_\tau \Phi(S_\tau)$ at exercise time $\tau$. The amortization immediately shows up as an additional discounting term (i.e., an updated effective risk-free rate $(r_t+q_t)_{t \geq 0}$); with this augmented risk-free rate, the amortization also appears as a drag on the modified drift to the underlying asset (i.e., an effective dividend yield augmented by $(q_t)_{t \geq 0}$). In this way, in Lemma~\ref{lemma:ampo}, we demonstrate that AmPO valuation perfectly coincides with (vanilla) perpetual American options under modified parameters. 
\begin{lemma}\label{lemma:ampo}
Let the risky asset price process be given by $(S_t)_{t \geq 0}$ with dividend rate $(\delta_t)_{t \geq 0}$.
The AmPO with amortization rate $(q_t)_{t \geq 0}$ is priced (per unit of notional) identically to the perpetual American option with the same payoff, effective risk-free rate $(r_t+q_t)_{t \geq 0}$, and dividend yield $(\delta_t+q_t)_{t \geq 0}$.
\end{lemma}
\begin{proof}
Consider the money-market account $B_t := \exp(\int_0^t r_s ds)$ for any time $t$ and assume the risky asset pays a (possibly time-varying) dividend yield $(\delta_t)_{t \geq 0}$ in the sense that the discounted, dividend-reinvested price process is a $\Q$-local martingale, i.e., $(B_t^{-1}\exp(\int_0^t \delta_s ds) S_t)_{t \geq 0}$ is a $\Q$-local martingale.
Furthermore, define the effective money-market account $\tilde{B}_t := \exp(\int_0^t (r_s + q_s) ds)$. Immediately, as $\tilde{B}_t^{-1} \exp(\int_0^t (\delta_s + q_s) ds) S_t = B_t^{-1} \exp(\int_0^t \delta_s ds) S_t$, we note that $(S_t)_{t \geq 0}$ with effective dividend yield $(\delta_t + q_t)_{t \geq 0}$ is a $\Q$-local martingale with discounting by $(\tilde{B}_t)_{t \geq 0}$.
Now consider the AmPO value at time $t = 0$ provided in~\eqref{eq:exercise}, i.e., $V_0 = \sup_{\tau \in \mathcal{T}} \E^\Q[\tilde{B}_\tau^{-1}\Phi(S_\tau)]$. By inspection, the stopping problem defining $V_0$ is exactly the perpetual American option valuation problem in a market with risk-free rate $(r_t + q_t)_{t \geq 0}$ and dividend yield $(\delta_t + q_t)_{t \geq 0}$.
\end{proof}

Lemma~\ref{lemma:ampo} provides a direct methodology to study the risk-neutral valuation of AmPOs in general settings. By solving the equivalent perpetual American option valuation problem, the valuation and exercise boundary for AmPOs can be determined. We refer the interested reader to, e.g., \cite{boyarchenko2000option,boyarchenko2002perpetual} for theory on pricing these options under general L\'{e}vy processes and to~\cite{al2024perpetual} for a methodology for such problems under jump-diffusion processes.
In Section~\ref{sec:bs-const} below, we will detail pricing and hedging in the Black-Scholes framework under a constant amortization rate $q_t \equiv q > 0$ to provide a simple setting with closed form solutions.

Though we find that AmPOs can be priced as perpetual American options on dividend-paying assets with modified market parameters, we note that these still function as CI variants. In Proposition~\ref{prop:ampo-CI}, we verify that the AmPO valuation problem is identical to that of the CI option with $c_t = q_t V_t$ (assuming constant claimable notional) under regularity conditions and if the amortization rate is exogenously defined from the option value; while not required for the AmPO valuation results, it provides sufficient conditions under which the self-referential CI option coincides with the intuition of Section~\ref{sec:ampo}. Such results are akin to those found in pricing credit risk~\cite{duffie1999modeling}.
\begin{proposition}\label{prop:ampo-CI}
Fix an exogenous amortization rate process $(q_t)_{t \geq 0}$.
Define the operator $\mathcal{C}$ on $(U_t)_{t \geq 0} \in \mathbb{V} := \left\{(U_t)_{t \geq 0} \geq 0 \text{ adapted c\`adl\`ag} \; \bigg| \; \left(e^{-\int_0^s r_u du} U_s\right)_{s \in [0,t]} \text{ is of class $D$}, \; \int_0^t e^{-\int_0^s r_u du} q_s U_s ds \in L^1 \, \forall t \geq 0\right\}$ by 
\[\mathcal{C}(U)_t := \esssup_{\tau \in \mathcal{T}, \tau \geq t} \E^\Q\left[e^{-\int_t^\tau r_s ds} \Phi(S_\tau) - \int_t^\tau e^{-\int_t^s r_u du} q_s U_s ds \; \bigg| \; \mathcal{F}_t\right].\footnote{Because immediate exercise is admissible, $\mathcal{C}(U)_t \geq \Phi(S_t) \geq 0$. Thus adding an outside option to lapse (i.e., replacing $\mathcal{C}(U)_t$ by $\mathcal{C}(U)_t \vee 0$) would not change the operator.}\]
Assume the per unit of notional AmPO value process $(V_t)_{t \geq 0} \in \mathbb{V}$.
\begin{enumerate}
\item\label{prop:ampo-CI-1} If $\mathcal{C}$ admits a fixed point $U \in \mathbb{V}$ then $U = V$ (i.e., if $V = \mathcal{C}(V)$ then it is the unique fixed point);
\item\label{prop:ampo-CI-2} if $(S_t)_{t \geq 0}$ has continuous paths, $\Phi$ is a continuous mapping, $\lim_{t \to \infty} e^{-\int_0^t r_s ds} \Phi(S_t) = 0$ a.s.,\linebreak $\E^\Q[\sup_{t \geq 0} e^{-\int_0^t r_s ds} \Phi(S_t)] < \infty$ and $V \in \mathbb{V}^* := \{U \in \mathbb{V} \; | \; \int_0^\infty e^{-\int_0^s r_u du} q_s U_s ds \in L^1(\Q)\}$ then $V = \mathcal{C}(V)$.
\end{enumerate}
\end{proposition}
\begin{proof}
For clarity we will break this proof into three parts. First, we will show that the per unit of notional value process $V \geq \mathcal{C}(V)$. Second, given this result, we will show that $U \in \mathbb{V}$ such that $U = \mathcal{C}(U)$ must coincide with $V$. Finally, under the added condition that $V \in \mathbb{V}^*$, we will show that $V \leq \mathcal{C}(V)$ which, combined with the first part of the proof, will guarantee that $V = \mathcal{C}(V)$. For these purposes, define the conditional payoff mapping $\bar{\mathcal{C}}^\tau$ for a stopping time $\tau \in \mathcal{T}$ such that 
\[\bar{\mathcal{C}}^\tau(U)_t := \E^\Q\left[e^{-\int_{t\wedge\tau}^\tau r_s ds} \Phi(S_\tau) - \int_{t\wedge\tau}^\tau e^{-\int_{t\wedge\tau}^s r_u du} q_s U_s ds \; \big| \; \mathcal{F}_{t\wedge\tau}\right]\] 
for any $U \in \mathbb{V}$. 
Define the discount factor $D_t := \exp(-\int_0^t r_s ds)$ and the cumulative amortization factor $Q_t := \exp(-\int_0^t q_s ds)$. From the definition of the American option value in Lemma~\ref{lemma:ampo}, the amortized, discounted value process $M_t := D_t Q_t V_t$ is the Snell envelope of $(D_t Q_t \Phi(S_t))_{t \geq 0}$ and, thus, is a $\Q$-supermartingale for all $t$. 
\begin{enumerate}
\item Consider $\bar{\mathcal{C}}^\tau(V)_t$ for an arbitrary stopping time $\tau \in \mathcal{T}$ (assuming $\tau \geq t$ a.s.\ without loss of generality):
\begin{align*}
\bar{\mathcal{C}}^\tau&(V)_t = \frac{1}{D_t} \E^\Q\left[D_\tau \Phi(S_\tau) - \int_t^\tau D_s q_s V_s ds \, \bigg| \, \mathcal{F}_t\right] 
 = \frac{1}{D_t} \E^\Q\left[D_\tau \Phi(S_\tau) - \int_t^\tau q_s Q_s^{-1} M_s ds \, \bigg| \, \mathcal{F}_t \right] \\ 
 &= \frac{1}{D_t} \E^\Q\left[D_\tau \Phi(S_\tau) - \int_t^\tau M_s d(Q_s^{-1}) \, \bigg| \, \mathcal{F}_t \right] 
 = \frac{1}{D_t} \E^\Q\left[D_\tau \Phi(S_\tau) - Q_\tau^{-1} M_{\tau} + Q_t^{-1} M_t + \int_t^\tau Q_s^{-1} dM_s \, \bigg| \, \mathcal{F}_t\right].
\end{align*}
Recall that $Q_\tau^{-1} M_\tau = D_\tau V_\tau$ and $Q_t^{-1} M_t = D_t V_t$. Substituting these back into the expectation:
\[
\bar{\mathcal{C}}^\tau(V)_t = V_t + \frac{1}{D_t} \underbrace{\E^\Q\left[D_\tau (\Phi(S_\tau) - V_\tau) \, \bigg| \, \mathcal{F}_t\right]}_{\leq 0} + \frac{1}{D_t} \underbrace{\E^\Q\left[\int_t^\tau Q_s^{-1} dM_s \, \bigg| \, \mathcal{F}_t\right]}_{\leq 0} \leq V_t.
\]
The first term is non-positive because the American option value always exceeds the immediate payoff ($V_\tau \geq \Phi(S_\tau)$). 
To prove that the second term is non-positive, we first note that $X_T := \int_t^T Q_s^{-1} dM_s$ (for $T \geq t$) is a $\Q$-local supermartingale (as $(M_s)_{s \geq t}$ is a $\Q$-supermartingale and $(Q_s^{-1})_{s \geq t}$ is positive and of finite variation). Let $(\sigma_k)_{k \in \mathbb{N}}$ be a localizing sequence of stopping times such that $(X_T^k := X_{T \wedge \sigma_k})_{T \geq t}$ are $\Q$-supermartingales for any $k$. Furthermore, because $Q_T^{-1}M_T \geq 0$ for any $T \geq 0$, we note that $X_\tau^k = Q_{\tau\wedge\sigma_k}^{-1} M_{\tau\wedge\sigma_k} - Q_t^{-1} M_t - \int_t^{\tau\wedge\sigma_k} D_s q_s V_s ds \geq -Q_t^{-1} M_t - \int_t^\tau D_s q_s V_s ds \in L^1$ by the assumed integrability conditions. Therefore, applying Fatou's lemma, we find:
\[\E^\Q\left[\int_t^\tau Q_s^{-1} dM_s \, \bigg| \, \mathcal{F}_t\right] \leq \liminf_{k \to \infty} \E^\Q\left[X_\tau^k \, | \, \mathcal{F}_t\right] \leq \liminf_{k \to \infty} X_t^k = 0.\]
\item Let $U \in \mathbb{V}$ be a fixed point of $\mathcal{C}$.
Since immediate exercise is admissible, $U_t = \mathcal{C}(U)_t \geq \Phi(S_t)$ for any $t \geq 0$.
We now show that $U_t \geq V_t$ for all $t \geq 0$. Taking advantage of $U = \mathcal{C}(U)$ and the dynamic programming principle, $N_t := D_t U_t - \int_0^t D_s q_s U_s ds$ defines a $\Q$-supermartingale. Moreover, $D_t Q_t U_t = U_0 + \int_0^t Q_s dN_s$ for all times $t \geq 0$. Since $(Q_t)_{t \geq 0}$ is positive, predictable, and of finite variation, the right-hand side is a $\Q$-local supermartingale; since the left-hand side is non-negative, Fatou's lemma implies $(D_t Q_t U_t)_{t \geq 0}$ is a $\mathbb{Q}$-supermartingale.
Therefore, because $U_t \geq \Phi(S_t)$, this supermartingale dominates the discounted payoff process, i.e., $D_t Q_t U_t \geq D_t Q_t \Phi(S_t)$ for any time $t \geq 0$. In particular, for any time $t \geq 0$, 
\[D_t Q_t U_t \geq \esssup_{\tau \in \mathcal{T}, \tau \geq t} \E^\Q[D_\tau Q_\tau U_\tau \, | \, \mathcal{F}_t] \geq \esssup_{\tau \in \mathcal{T}, \tau \geq t} \E^\Q[D_\tau Q_\tau \Phi(S_\tau) \, | \, \mathcal{F}_t] = D_t Q_t V_t.\]
Finally, following step 1 of this proof and because $\mathcal{C}$ is order reversing (by inspection), $U_t = \mathcal{C}(U)_t \leq \mathcal{C}(V)_t \leq V_t$ for all times $t \geq 0$. Therefore, combining with the proven $U_t \geq V_t$, it must follow that $U_t = V_t$ at all times $t \geq 0$.
\item Under the imposed continuity and integrability conditions, by \cite[Theorem D.12]{karatzas1998methods}, there exists an optimal stopping time $\tau_t^* \geq t$ such that $(M_{s \wedge \tau_t^*})_{s \geq t}$ is a $\Q$-martingale and $V_{\tau_t^*} = \Phi(S_{\tau_t^*})$ on $\{\tau_t^* < \infty\}$.
Truncating this stopping time $\tau_t^T := \tau_t^* \wedge T$ at $T \geq t$, we recover $\mathcal{C}(V)_t \geq \bar{\mathcal{C}}^{\tau_t^T}(V)_t$. In particular, by the assumed integrability conditions, the dominated convergence theorem implies $\mathcal{C}(V)_t \geq \frac{1}{D_t} \E^\Q\left[D_{\tau_t^*} \Phi(S_{\tau_t^*}) - \int_t^{\tau_t^*} D_s q_s V_s ds \; \bigg| \; \mathcal{F}_t\right]$ where $D_{\tau_t^*} \Phi(S_{\tau_t^*}) = 0$ on $\{\tau_t^* = \infty\}$. Furthermore, using the martingale property of $(M_{s \wedge \tau_t^*})_{s \geq t}$, by Tonelli's theorem and the tower property,
\begin{align*}
\E^\Q\left[\int_t^{\tau_t^*} D_s q_s V_s ds \; \bigg| \; \mathcal{F}_t\right] &= \E^\Q\left[D_{\tau_t^*} \Phi(S_{\tau_t^*}) \int_t^{\tau_t^*} q_s \frac{Q_{\tau_t^*}}{Q_s} ds \; \bigg| \; \mathcal{F}_t\right] 
= \E^\Q\left[D_{\tau_t^*} \Phi(S_{\tau_t^*}) \left(1 - \frac{Q_{\tau_t^*}}{Q_t}\right) \; \bigg| \; \mathcal{F}_t\right].
\end{align*}
Therefore, $\mathcal{C}(V)_t \geq  \E^\Q\left[\frac{D_{\tau_t^*}Q_{\tau_t^*}}{D_t Q_t}\Phi(S_{\tau_t^*}) \; \bigg| \; \mathcal{F}_t\right] = \E^\Q\left[e^{-\int_t^{\tau_t^*} (r_s + q_s) ds} \Phi(S_{\tau_t^*}) \; \bigg| \; \mathcal{F}_t\right] = V_t$.
\end{enumerate}
\end{proof}

Notably, the requirement in Proposition~\ref{prop:ampo-CI} of an exogenously defined amortization rate process $(q_t)_{t \geq 0}$ (i.e., one that is not a direct function of $V$) is not satisfied in standard CI option constructions (see Example~\ref{ex:CI}). 
In the following section, we will consider the specific case of a constant amortization rate $q_t \equiv q > 0$; this simplifies the setup for explicit valuation and satisfies the regularity conditions imposed in Proposition~\ref{prop:ampo-CI}\eqref{prop:ampo-CI-1}. While Proposition~\ref{prop:ampo-CI}\eqref{prop:ampo-CI-2} gives one sufficient condition for $V = \mathcal{C}(V)$, this condition fails for call options in the Black-Scholes setting. Therefore, under that simplified setting, Corollary~\ref{prop:bs-CI} directly establishes the fixed-point relation.

\subsection{Constant Amortization in the Black-Scholes Framework}\label{sec:bs-const}

\begin{table}[!ht]
\centering
\begin{tabular}{|l||c|c|}
\hline
 & \textbf{Call Option} & \textbf{Put Option} \\ \hline\hline
\textbf{Price} & $\frac{K}{\alpha_C-1}\left(\frac{(\alpha_C-1)S_0}{\alpha_C K}\right)^{\alpha_C}$ & $\frac{K}{1+\alpha_P}\left(\frac{\alpha_P K}{(1+\alpha_P)S_0}\right)^{\alpha_P}$ \\ \hline
\textbf{Delta} & $\left(\frac{(\alpha_C - 1) S_0}{\alpha_C K}\right)^{\alpha_C - 1}$ & $-\left(\frac{\alpha_P K}{(1+\alpha_P)S_0}\right)^{1+\alpha_P}$ \\ \hline
\textbf{Gamma} & $\frac{(\alpha_C-1)^2}{\alpha_C K} \left(\frac{(\alpha_C - 1) S_0}{\alpha_C K}\right)^{\alpha_C - 2}$ & $\frac{\alpha_P K}{S_0^2} \left(\frac{\alpha_P K}{(1+\alpha_P) S_0}\right)^{\alpha_P}$ \\ \hline
\textbf{Theta} & 0 & 0 \\ \hline
\textbf{Vega} & $\frac{4C_0}{\sigma}\log\left(\frac{(\alpha_C-1)S_0}{\alpha_C K}\right)\left(\frac{(\alpha_C-1)r-q}{(2\alpha_C-1)\sigma^2 + 2r}\right)$ & $\frac{4P_0}{\sigma}\log\left(\frac{(1+\alpha_P)S_0}{\alpha_P K}\right)\left(\frac{(1+\alpha_P)r+q}{(2\alpha_P+1)\sigma^2-2r}\right)$ \\ \hline
\end{tabular}
\caption{Summary table of AmPO risk-neutral valuation and Greeks for a constant notional exposure in the Black-Scholes framework under constant amortization with $S_0 \leq \bar{S}_C$ and $S_0 \geq \bar{S}_P$ for calls and puts respectively.
} \label{tab:greeks}
\end{table}

Within this section, we want to investigate the risk-neutral valuation of these AmPOs in the Black-Scholes framework under constant amortization. To do so, we follow Assumptions~\ref{ass:gbm} throughout this subsection.
\begin{assumption}\label{ass:gbm}
Consider a complete market with constant risk-free rate $r \geq 0$ and risky asset with price process $(S_t)_{t \geq 0}$.
Let $\Q \sim \P$ denote the unique risk-neutral measure for this market and consider the $\Q$-Brownian motion (w.r.t.\ $(\F_t)_{t \geq 0}$) $(W_t)_{t \geq 0}$.
For the remainder of this work, we will assume that $(S_t)_{t \geq 0}$ follows a geometric Brownian motion under $\Q$:
\[dS_t = S_t (r dt + \sigma dW_t)\]
with initial price $S_0 > 0$ and volatility $\sigma > 0$.
Furthermore, consider AmPOs with constant amortization rates $q_t \equiv q > 0$.
\end{assumption}

\subsubsection{Pricing and Hedging}\label{sec:bs-const-price}

Following Lemma~\ref{lemma:ampo}, the pricing of call and put AmPOs trivially follows from the results of, e.g., \cite[Chapter 26.2]{hull}. For this reason, we omit the proof for the pricing of call and put AmPOs.
\begin{corollary}[Call Option Pricing]\label{cor:call}
Consider a call AmPO with strike $K > 0$ and amortization rate $q > 0$. 
The optimal exercise boundary is $\bar S_C = \frac{\alpha_C K}{\alpha_C-1}$ for $\alpha_C = \sqrt{\left(\frac{r}{\sigma^2}+\frac{1}{2}\right)^2 + \frac{2q}{\sigma^2}} - \frac{r}{\sigma^2} + \frac{1}{2} > 1$. Under optimal execution, with $S_0 \leq \bar S_C$, the premium for this option is
\[C_0 = \frac{K}{\alpha_C-1}\left(\frac{(\alpha_C - 1)S_0}{\alpha_C K}\right)^{\alpha_C}.\]
\end{corollary}
\begin{corollary}[Put Option Pricing]\label{cor:put}
Consider a put AmPO with strike $K > 0$ and amortization rate $q > 0$. 
The optimal exercise boundary is $\bar S_P = \frac{\alpha_P K}{1+\alpha_P}$ for $\alpha_P = \sqrt{\left(\frac{r}{\sigma^2}+\frac{1}{2}\right)^2 + \frac{2q}{\sigma^2}} + \frac{r}{\sigma^2} - \frac{1}{2} > 0$. Under optimal execution, with $S_0 \geq \bar S_P$, the premium for this option is
\[P_0 = \frac{K}{1+\alpha_P}\left(\frac{\alpha_P K}{(1+\alpha_P)S_0}\right)^{\alpha_P}.\]
\end{corollary}

While Corollaries~\ref{cor:call} and~\ref{cor:put} provide explicit expressions for the AmPO value process, we now verify that these are the unique value processes of the self-referential CI variant.
\begin{corollary}\label{prop:bs-CI}
Consider either the call or put AmPO with strike $K > 0$ and amortization rate $q > 0$. 
The AmPO value process is the unique fixed point of $\mathcal{C}$ in $\mathbb{V}$. 
\end{corollary}
\begin{proof}
Consider the call and put AmPOs, i.e., $\Phi(S) = (S-K)^+$ or $\Phi(S) = (K-S)^+$ respectively. 
The AmPO value process satisfies $V \in \mathbb{V}$ under Assumption~\ref{ass:gbm} due to the bounds for puts ($0 \leq V_t \leq K$) and calls ($0 \leq V_t \leq S_t$, together with the discounted martingale construction of $(S_t)_{t \geq 0}$). To conclude this proof, let $\tau^*$ denote the exercise-boundary hitting time associated with the appropriate option and set $\tau^T := \tau^* \wedge T$. 
Due to the Markovian property of the geometric Brownian motion, herein we only provide the fixed point relation at time $t = 0$ while all other times follow comparably. 
Note that $V_0 = \E^\Q[e^{-(r+q)\tau^T} V_{\tau^T}]$ for any $T > 0$. Using the notation from the proof of Proposition~\ref{prop:ampo-CI}, by way of Tonelli's theorem and the tower property:
\begin{align*}
\bar{\mathcal{C}}^{\tau^T}(V)_0 &= \E^\Q\left[e^{-r\tau^T} \Phi(S_{\tau^T}) - \int_0^{\tau^T} e^{-rs} q \E^\Q\left[e^{-(r+q)(\tau^T-s)} V_{\tau^T} \, \big| \, \mathcal{F}_s\right] ds \right] \\
&= \E^\Q\left[e^{-r\tau^T} \Phi(S_{\tau^T}) - e^{-r\tau^T} V_{\tau^T} \int_0^{\tau^T} q e^{-q (\tau^T-s)} ds \right] \\
&= V_0 + \E^\Q\left[e^{-rT} (\Phi(S_T) - V_T)\mathbb{I}_{\{\tau^T = T\}}\right].
\end{align*}
Note that the residual term on $\{\tau^T = T\}$ limits to 0 a.s.\ for both call and put AmPOs. Moreover, this residual term is dominated in both cases by integrable random variables (by $\bar{S}_C$ and $K$ for calls and puts respectively). As such, by the dominated convergence theorem, the expectation of this residual term converges to zero.
Therefore, letting $T \to \infty$, we find $\mathcal{C}(V)_0 \geq V_0$ which completes the proof when combined with steps 1 and 2 of the proof of Proposition~\ref{prop:ampo-CI}.
\end{proof}

As call and put AmPOs can be valued via an equivalence relation with perpetual American options as provided in Lemma~\ref{lemma:ampo} and the above explicit representations, we can follow standard results to determine the dependence of these structures on the underlying system parameters. Such results are summarized in Table~\ref{tab:greeks}.
\begin{remark}
Via the relation to perpetual American options, we consider the Greeks of the valuation for a \emph{constant} notional exposure, i.e., dependence of the explicit premia $C_0$ and $P_0$ for calls and puts, respectively. However, we note that the realized exposure is amortizing for an investor that is either long or short AmPOs. Therefore, in practice, the realized Greeks should be discounted by the amortization factor appropriately. 

This is especially prominent for Theta which, typically, refers to the dependence of the premium on the ``time to maturity''. Because AmPOs are perpetual options, the risk-neutral price exhibits zero explicit time decay (e.g., $\frac{\partial C_0}{\partial t} = 0$ for call options). 
However, due to the amortization, the holder or underwriter of such options will, in fact, have an \emph{economic exposure} that depends on time. This decay is directly computable as $-q C_0 < 0$ or $-q P_0 < 0$ for call and put options, respectively. It is exactly this non-zero Theta(-like) decay that distinguishes AmPOs from vanilla perpetual options and creates an effective realized time exposure.
\end{remark}

We conclude this discussion of risk-neutral valuation by considering how the amortization $q$ can mimic an effective maturity date $T_{\text{eff}}$. That is, for a given amortization rate $q$, we define the effective maturity $T_{\text{eff}} := \inf\{T \geq 0 \; | \; V_0(S_0;r,\sigma,q) \leq V^{A}(S_0;T,r,\sigma)\}$ as the maturity $T$ for a dated American option with the same market dynamics ($S_0,r,\sigma$) which matches the AmPO price.\footnote{The effective maturity $T_{\text{eff}} \in (0,\infty)$ for any $q \in (0,\infty)$ is such that $V_0(S_0;r,\sigma,q) = V^{A}(S_0;T_{\text{eff}},r,\sigma)$ as $V^{A}(S_0;0,r,\sigma) = \Phi(S_0) < V_0(S_0;r,\sigma,q) < \lim_{T \to \infty} V^{A}(S_0;T,r,\sigma)$ by Lemma~\ref{lemma:ampo}; see also Remark~\ref{rem:limit}.} 
\begin{example}\label{ex:maturity}
\begin{figure}[t]
\centering
\begin{subfigure}[t]{0.4\textwidth}
\centering
\includegraphics[width=\textwidth]{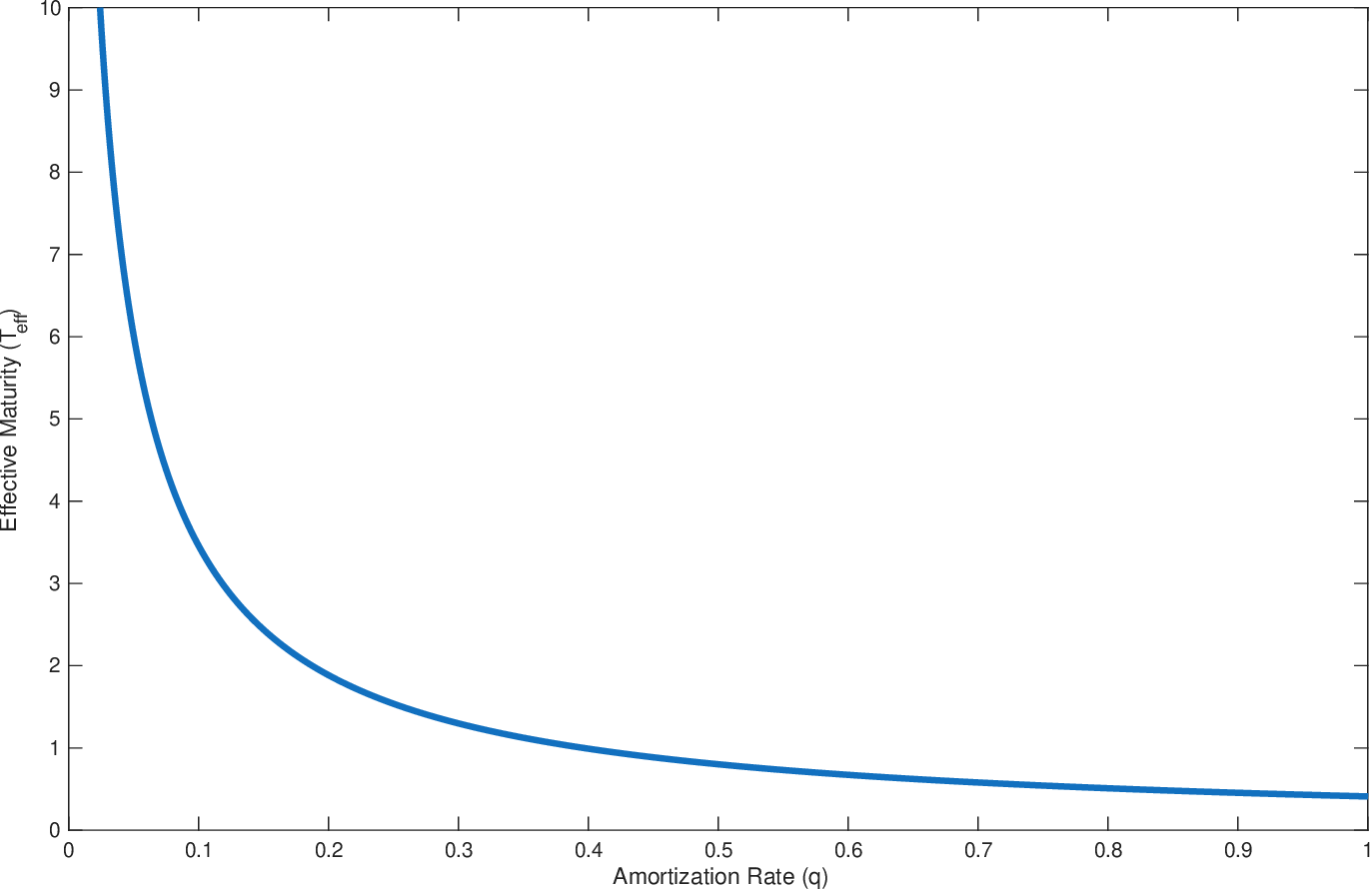}
\caption{The effective maturity date $T_{\text{eff}}$ as a function of the amortization rate $q$.}
\label{fig:maturity-T}
\end{subfigure}
~~
\begin{subfigure}[t]{0.4\textwidth}
\includegraphics[width=\textwidth]{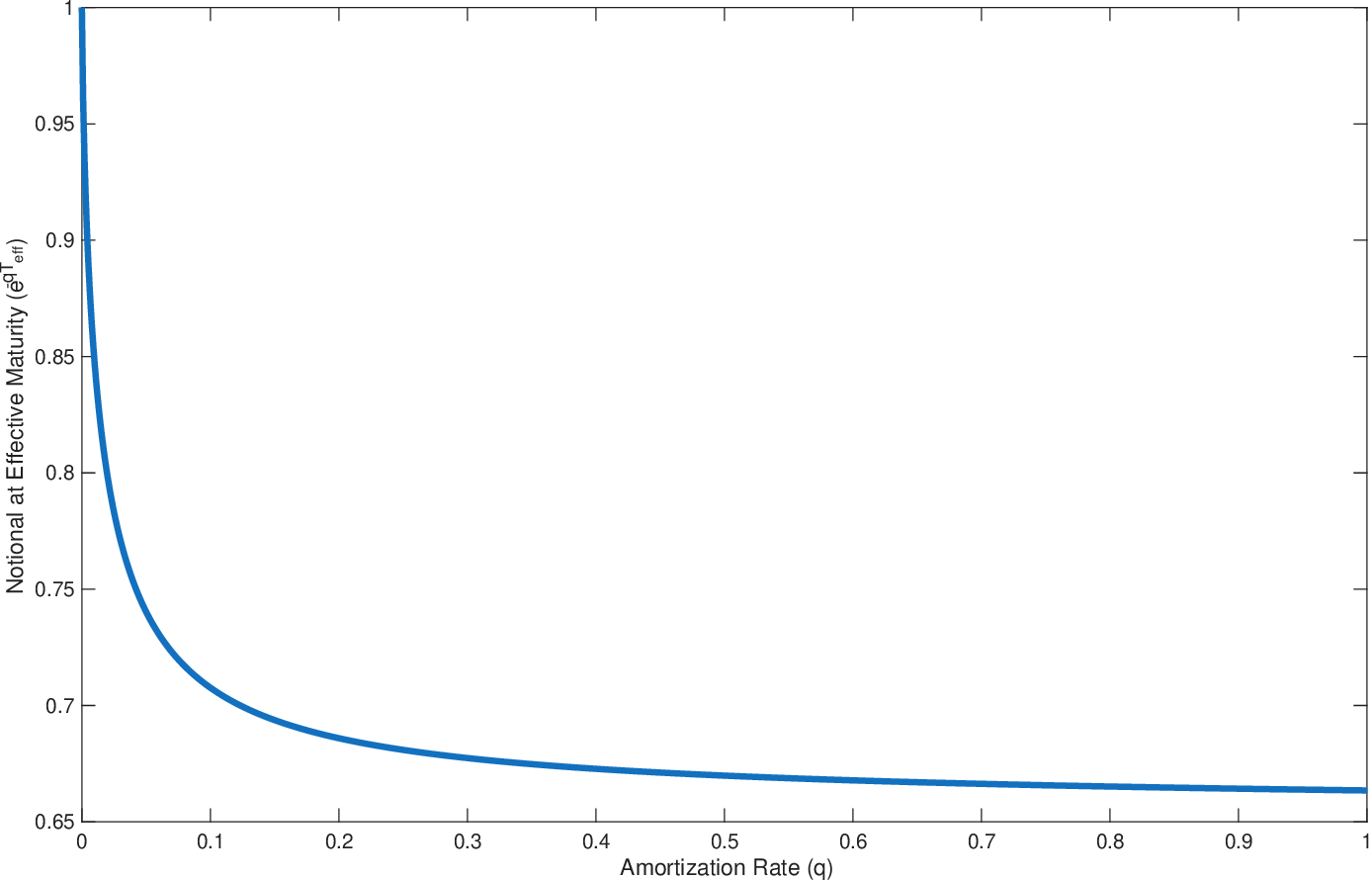}
\caption{The effective notional ($e^{-qT_{\text{eff}}}$) of the AmPO with amortization rate $q$ at its effective maturity date $T_{\text{eff}}$.}
\label{fig:maturity-notional}
\end{subfigure}
~~
\begin{subfigure}[t]{0.4\textwidth}
\centering
\includegraphics[width=\textwidth]{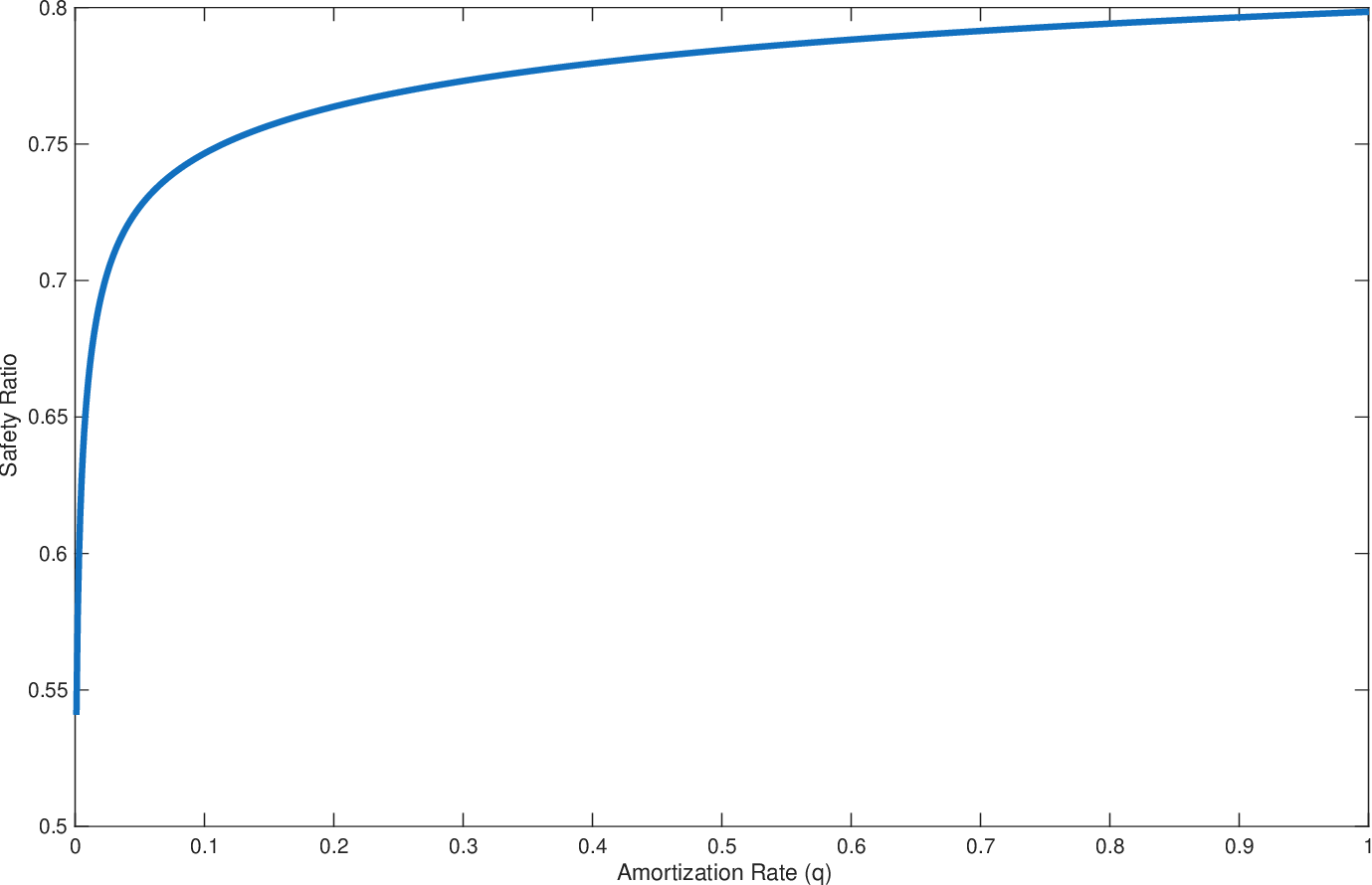}
\caption{The ratio of the AmPO Gamma against its effectively dated peer.}
\label{fig:greeks-gamma}
\end{subfigure}
~~
\begin{subfigure}[t]{0.4\textwidth}
\includegraphics[width=\textwidth]{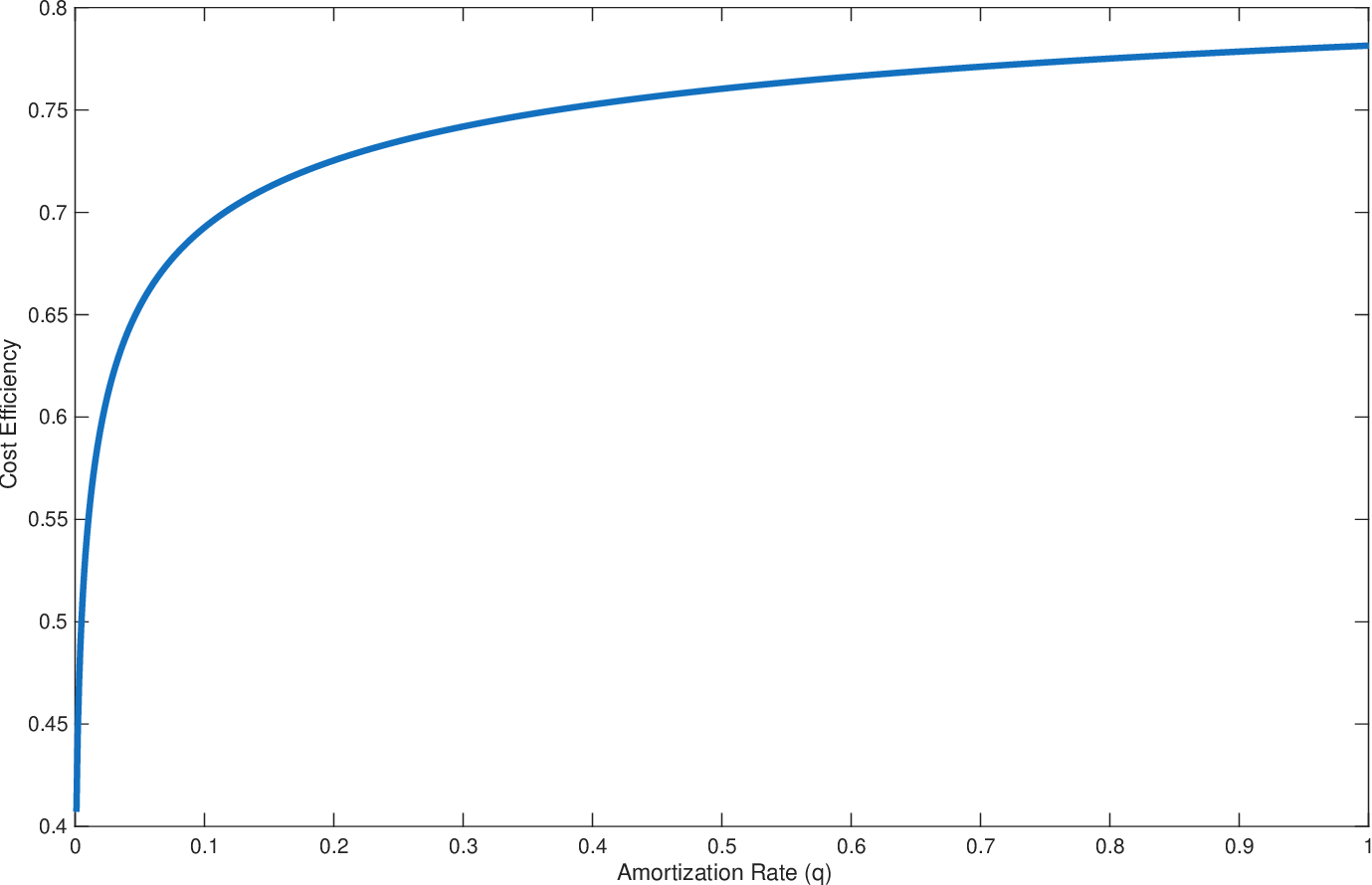}
\caption{The ratio of the AmPO (economic) Theta against its effectively dated peer.}
\label{fig:greeks-theta}
\end{subfigure}
\caption{Example~\ref{ex:maturity}: Comparison of AmPO and dated at-the-money call options.}
\label{fig:maturity}
\end{figure}
Consider an asset price following Assumption~\ref{ass:gbm} with risk-free rate $r = 5\%$, volatility $\sigma = 50\%$, and initial price $S_0 = 100$. Consider the at-the-money call option (i.e., $K = 100$) with amortization rate $q$. 

In Figure~\ref{fig:maturity-T}, we plot the effective maturity $T_{\text{eff}}$ of the at-the-money (American) call option that has the same premium as the AmPO with amortization rate $q$.\footnote{As the assets under consideration here pay 0 dividends, American and European call options have the same premium.} As expected, the effective maturity $T_{\text{eff}}$ is decreasing in the amortization rate $q$. 
However, Figure~\ref{fig:maturity-notional} reveals that $T_{\text{eff}}$ decreases at a slower rate than $q$ increases. Consequently, the effective notional (the claimable notional $e^{-qT_{\text{eff}}}$ remaining at the effective maturity $T_{\text{eff}}$) is decreasing in the amortization rate $q$. Crucially, this effective notional illustrates that this effective maturity is solely a pricing equivalence and does not represent even an approximate termination of the AmPO; unlike the dated option which expires at $T_{\text{eff}}$, the AmPO retains a substantial portion of its notional at that time.
For the selected parameters, even at high amortization rates ($q \approx 1$), the AmPO has over 65\% of its original notional value remaining at the time that the equivalent dated option would expire.

Though the premiums of the AmPO and its effectively dated counterpart coincide, their sensitivities do not because the underlying stopping problems are distinct. For this reason, we now compare (select) Greeks of AmPOs against their effectively dated counterparts.
First, in Figure~\ref{fig:greeks-gamma}, we plot the ratio of the AmPO Gamma to the Gamma of the effectively dated (American) call option, i.e., $\Gamma_0(S_0;r,\sigma,q)/\Gamma^{A}(S_0;T_{\text{eff}},r,\sigma)$. 
The lower this ratio, the more stable the option price is to fluctuations in the underlying; for this reason, we call this the ``safety ratio.'' Notably, this safety ratio is increasing in the amortization rate, but stays below 80\% for even extremely high amortization rates (e.g., $q \approx 1$).
Similarly, in Figure~\ref{fig:greeks-theta}, we consider the ``cost efficiency'' of AmPOs by computing the ratio between the economic Theta (i.e., $-q C_0$) from amortization against the Theta of the effectively dated option, i.e., $-q C_0 / \Theta^{A}(S_0;T_{\text{eff}},r,\sigma)$. For a cost efficiency of less than 1, the AmPO's loss in notional from amortization is less than the loss in the effectively dated option from the passage of time. Here, again, we find that the dated options lose value faster even at high amortization rates with a maximum observed efficiency ratio (at $q = 1$) of under 80\%.
\end{example}

\subsubsection{Comparative Statics on Amortization Rate}\label{sec:bs-amort}

As the amortization rate $q$ appears directly within both the effective risk-free rate ($r+q$) and effective dividend yield ($q$) of the equivalently valued perpetual American option, the choice of this contractual parameter can have a strong influence on the valuation and Greeks of AmPOs. 
However, since we have proposed AmPOs to maintain a constant $q$ throughout the lifetime of the option, we do not want to consider this dependence in the same way as the Greeks (i.e., for hedging purposes), but rather as a way to understand the implications and tradeoffs of choosing the amortization rate $q$ for a contract. In doing so, we will focus on two aspects of the AmPO contracts: the risk-neutral valuation and the optimal exercise boundary. 
Throughout this section, we set $\bar\alpha := (\alpha_C + \alpha_P)/2 > 0$ for notational simplicity.

\begin{assumption}
For simplicity of notation, throughout this section we consider $S_0 \leq \bar S_C$ for call options and $S_0 \geq \bar S_P$ for put options.
\end{assumption}

First, as follows from~\cite[Proposition 5]{quah2009comparative} when applied to~\eqref{eq:exercise} directly, increasing $q$ accelerates economic amortization, reducing the value for the holder and, therefore, also the premium. 
\begin{corollary}\label{cor:q-premium}
Consider call and put AmPOs with strike $K > 0$. The premium for both option types is decreasing in the amortization rate $q > 0$:
\[\frac{\partial C_0}{\partial q} = \frac{C_0}{\sigma^2\bar\alpha} \log\left(\frac{(\alpha_C-1)S_0}{\alpha_C K}\right) \leq 0 \quad \text{ and } \quad \frac{\partial P_0}{\partial q} = -\frac{P_0}{\sigma^2\bar\alpha} \log\left(\frac{(1+\alpha_P)S_0}{\alpha_P K}\right) \leq 0.\]
\end{corollary}
\begin{proof}
Consider the call AmPO; we omit the proof of put AmPOs as the arguments follow comparably.
Recall from Corollary~\ref{cor:call} that the premium of the call AmPO $C_0$ depends on the amortization rate $q$ solely through $\alpha_C$. Thus, through an application of the chain rule, we find $\frac{\partial C_0}{\partial q} = \frac{\partial C_0}{\partial\alpha_C}\frac{\partial \alpha_C}{\partial q}$. By differentiation, we find
$\frac{\partial C_0}{\partial \alpha_C} = C_0\log\left(\frac{(\alpha_C-1)S_0}{\alpha_C K}\right) \leq 0$ and 
$\frac{\partial\alpha_C}{\partial q} = \frac{1}{\sigma^2\bar\alpha} > 0$
thus completing the proof.
\end{proof}

\begin{remark}\label{rem:limit}
Within Corollary~\ref{cor:q-premium}, we found that the premium for both AmPO calls and puts are decreasing in the amortization rate $q$. In fact, it can easily be seen that the limiting behavior of these options is to vanilla perpetual American options as $q \searrow 0$ and to the intrinsic value of the option as $q \nearrow \infty$.
\end{remark}

As the option premium decreases in the amortization rate, the optimal exercise boundary is similarly decaying towards the strike price (i.e., decreasing for call options and increasing for put options). In this way, the optimal exercise becomes more conservative as the amortization accelerates.
\begin{corollary}\label{cor:q-exercise}
Consider call and put AmPOs with strike $K > 0$. The optimal exercise boundary is decreasing to $K$ for call options and increasing to $K$ for put options as $q$ increases, i.e.,
\[\frac{\partial \bar S_C}{\partial q} = -\frac{K}{\sigma^2(\alpha_C-1)^2\bar\alpha} \leq 0 \quad \text{ and } \quad \frac{\partial \bar S_P}{\partial q} = \frac{K}{\sigma^2(1+\alpha_P)^2\bar\alpha} \geq 0\]
with $\lim_{q \nearrow \infty} \bar S_C = K$ and $\lim_{q \nearrow \infty} \bar S_P = K$.
\end{corollary}
\begin{proof}
As with the proof of Corollary~\ref{cor:q-premium}, this result trivially follows from an application of the chain rule as $\bar S_C,\bar S_P$ depend only on $q$ through $\alpha_C,\alpha_P$ respectively.
\end{proof}

\begin{remark}\label{rem:statics-q}
Though we have provided Corollaries~\ref{cor:q-premium} and~\ref{cor:q-exercise} for the Black-Scholes setting only, an application of~\cite[Proposition 5]{quah2009comparative} indicates that the monotonicity of the dependence on the amortization rate $q$ found in those results holds also for general price processes.
\end{remark}

Finally, we highlight that the dependence of Vega on the amortization rate is more subtle. As $q$ increases, the AmPO becomes closer to intrinsic value, which reduces volatility sensitivity; however, when $q$ is small, increasing amortization can also move the contract away from the degenerate (non-amortizing) perpetual option limit, thereby increasing volatility sensitivity. Therefore, unlike the premium and exercise boundary, Vega need not be monotone in $q$.
In Example~\ref{ex:maturity}, we analyzed how AmPOs compare with dated options. In the following example, we wish to consider how selecting the amortization rate can be used to optimize a volatility-based strategy. 
\begin{example}\label{ex:amort}
Consider the setting of Example~\ref{ex:maturity}. Recall from Corollary~\ref{cor:q-premium} that the premium of an AmPO is decreasing in the amortization rate $q$, while the Vega per unit of notional need not be monotone in $q$. As the realized Vega of a position depends on the budget actually invested, we now consider how the ratio of the Vega to premium ($\nu_0/V_0$) behaves to consider a simplified optimization problem: consider an investor who wants to take as leveraged a volatility position as possible using AmPOs with a \$100 budget constraint. To do so, we consider three strategies: bullish volatility with at-the-money call options (Figure~\ref{fig:amort-call}); bearish volatility with at-the-money put options (Figure~\ref{fig:amort-put}); and neutral volatility with an at-the-money straddle (Figure~\ref{fig:amort-straddle}). 
As observed in Figure~\ref{fig:amort}, the total positional Vega depends on the interaction between the Vega $\nu_0$ and the option premium $V_0$.
Notably, though bullish and neutral strategies have increasing positional Vegas, the bearish strategy has a non-monotonic behavior coming from the non-trivial interaction of the premium and Vega. For the current parameter setting, the optimal bearish volatility strategy (which also provides the maximum positional Vega among all 3 strategies) is with an amortization rate of $q^* \approx 19.26\%$.
\begin{figure}[t]
\centering
\begin{subfigure}[t]{0.3\textwidth}
\centering
\includegraphics[width=\textwidth]{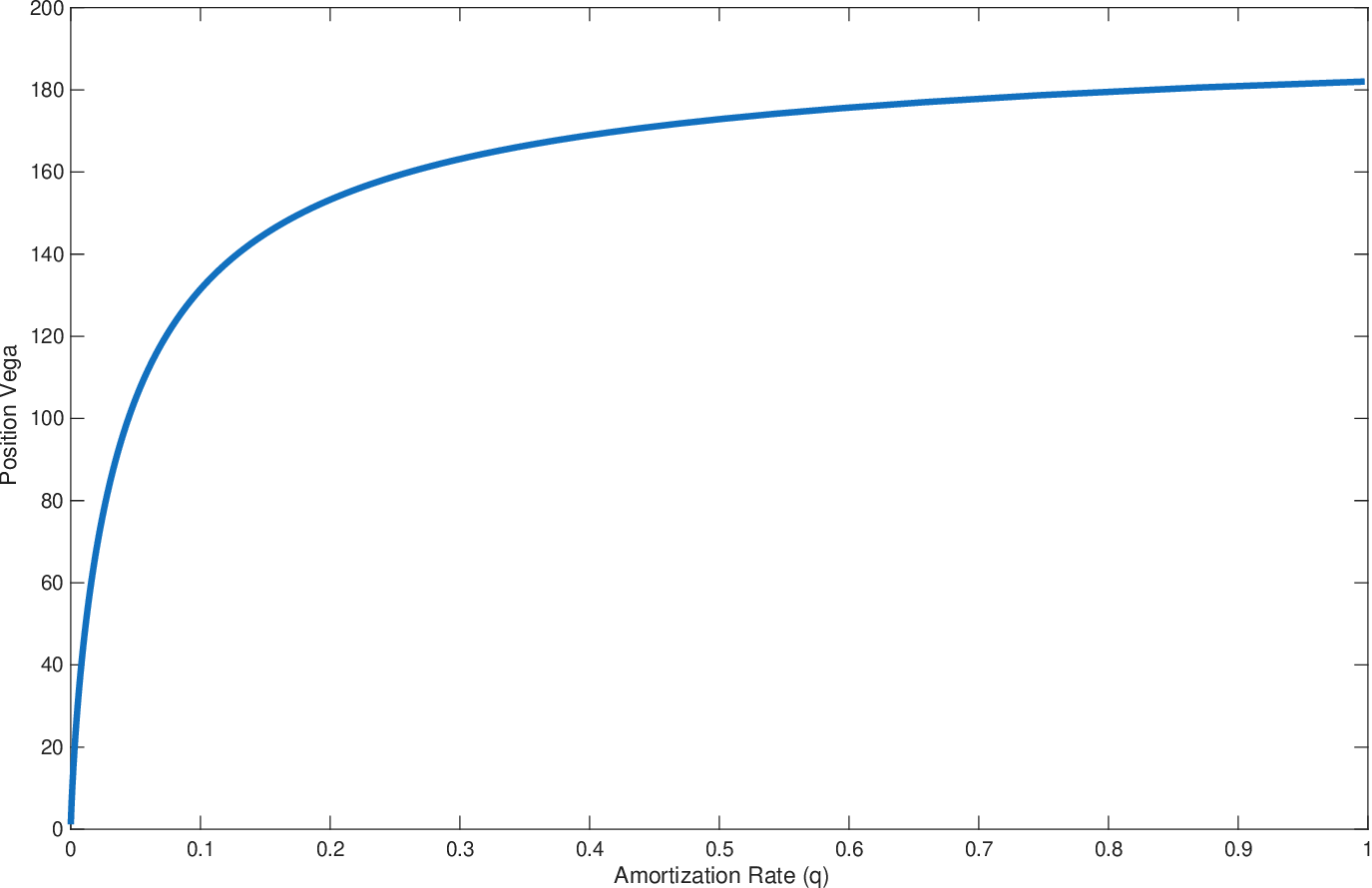}
\caption{Call options}
\label{fig:amort-call}
\end{subfigure}
~
\begin{subfigure}[t]{0.3\textwidth}
\includegraphics[width=\textwidth]{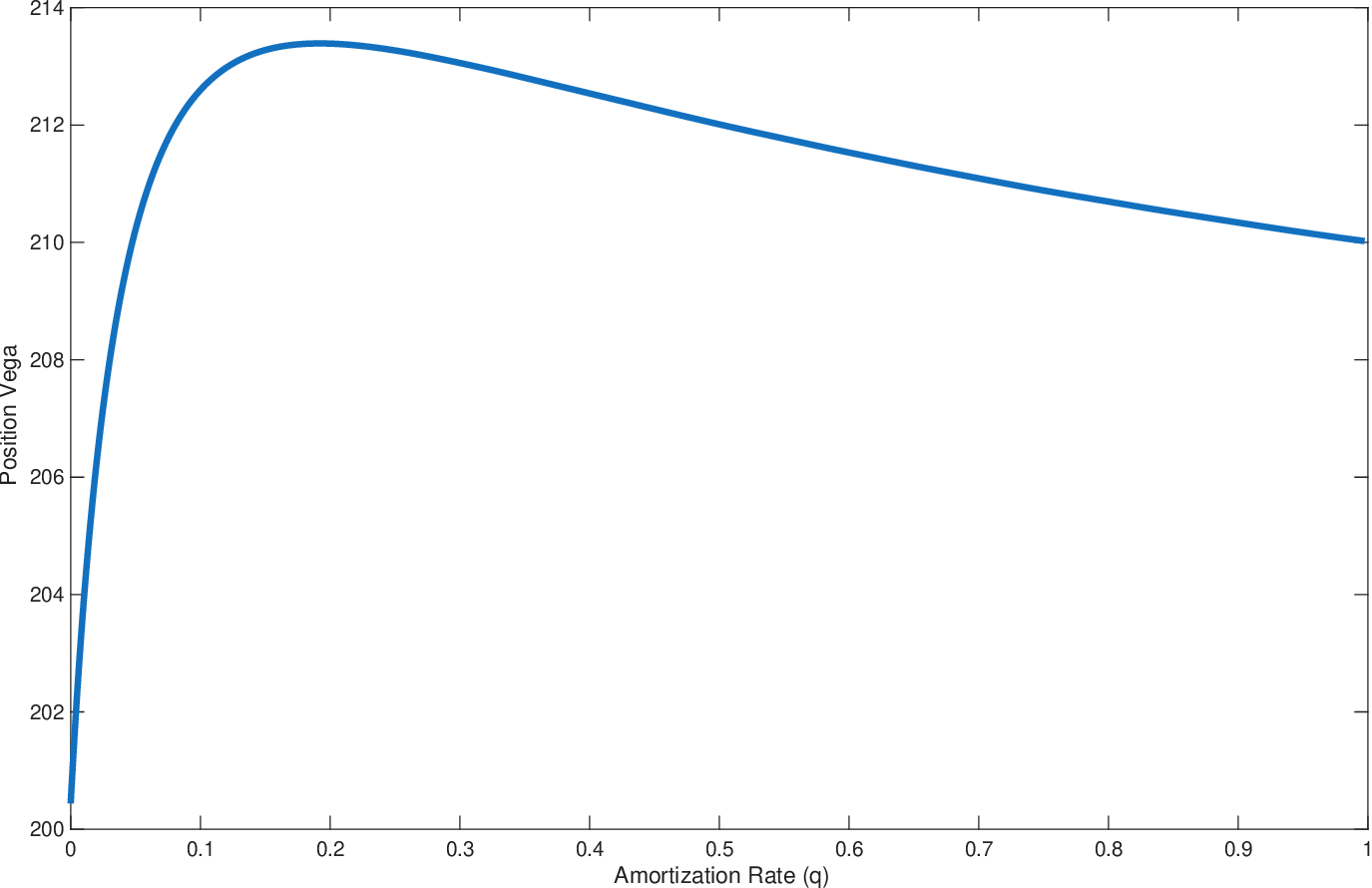}
\caption{Put options}
\label{fig:amort-put}
\end{subfigure}
~
\begin{subfigure}[t]{0.3\textwidth}
\includegraphics[width=\textwidth]{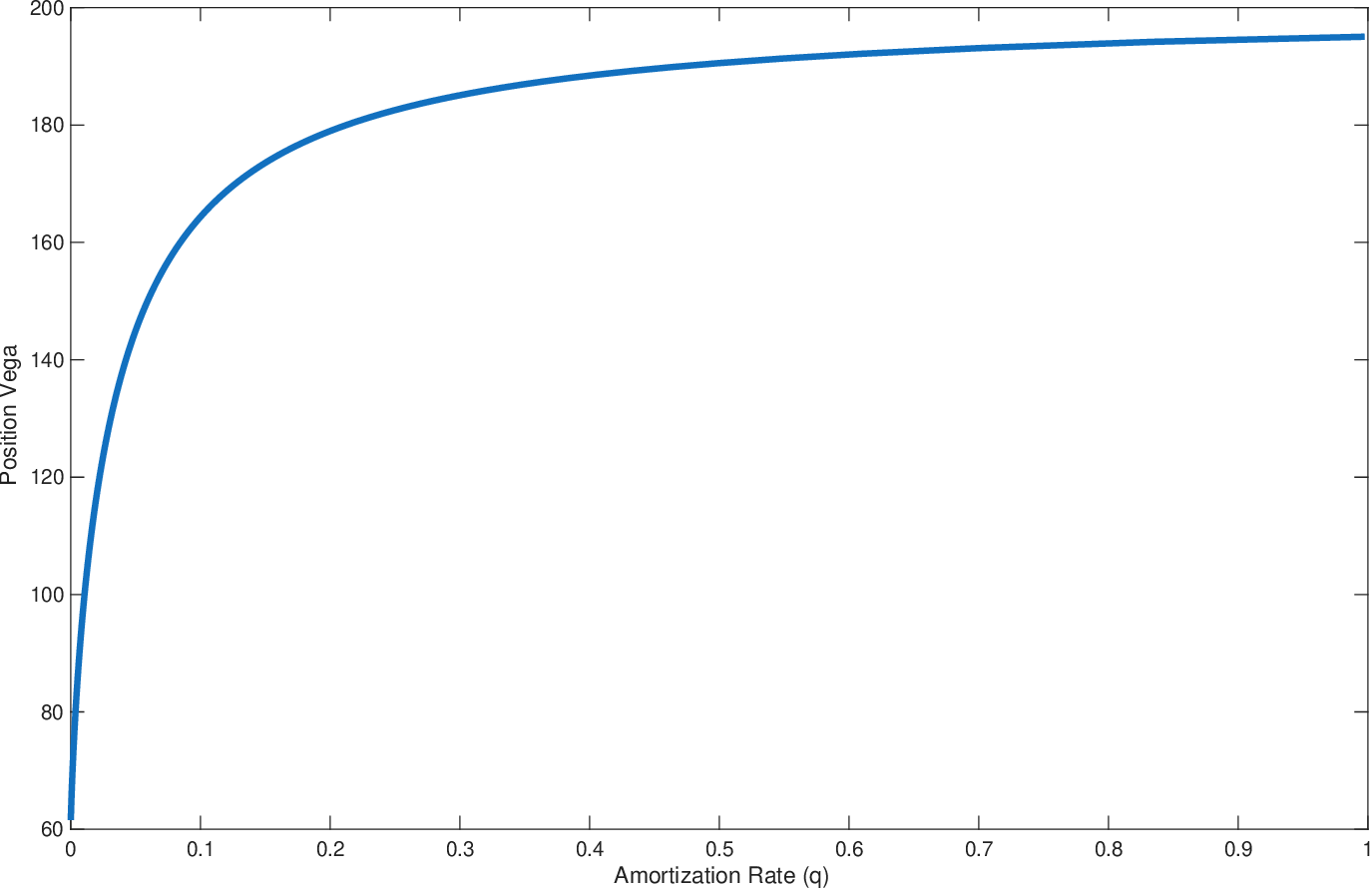}
\caption{Straddle}
\label{fig:amort-straddle}
\end{subfigure}
\caption{Example~\ref{ex:amort}: Vega of a \$100 position in call, put, and straddle strategies.}
\label{fig:amort}
\end{figure}
\end{example}

\section{Discussion and Applications}\label{sec:discussion}

\subsection{Setting the Amortization Rate}\label{sec:discussion-amort}

In Section~\ref{sec:bs-const}, we imposed a fixed $q_t \equiv q > 0$ amortization rate. While the general theory permits endogenous amortization (i.e., dependent on the option value), such a construction can introduce novel manipulations in AmPOs in practice. For instance, consider an AmPO with the fixed dollar installment rate $c_t \equiv c$, i.e., $q_t = c / V_t$, as is often taken for CI options (see, e.g., Example~\ref{ex:CI}). For a deep out-of-the-money AmPO ($V_t \searrow 0$), the amortization rate diverges to infinity ($q_t \nearrow \infty$) causing the open interest in the option series to crash to 0. This is the economically justified result \emph{if} the fair value is near 0; however, if the price is temporarily manipulated towards the 0 lower bound then this attack can effectively wipe out all long option positions. 
In particular, within decentralized finance where price manipulations are a constant threat, this endogenous amortization rate is particularly fragile.
In contrast, the fixed amortization rate $q_t \equiv q$ is \emph{not} subject to such manipulations as the amortization schedule is invariant to market prices. As such, though only a special case of AmPOs, we recommend setting the amortization rate as a fixed constant in practice.

Once the design space is restricted to the fixed amortization rate AmPOs, it still remains necessary to choose this rate $q > 0$. As AmPOs are exchange-tradable instruments, it is necessary to impose standard contracts in order to concentrate liquidity in specific markets; this is akin to setting a standard choice of maturity date for options (e.g., the third Friday of the month). As evidenced by Example~\ref{ex:amort}, it is clear that there does not exist a Pareto optimal amortization rate $q$ amongst all market participants. Instead, herein we suggest the market operator define a few canonical rates (e.g., 3 bps/day, 10 bps/day, and 50 bps/day) to segment market participants based on their implied maturities and risk tolerances (see, e.g., Figure~\ref{fig:maturity}). A formal study of optimal choice of amortization rates to maximize, e.g., total market volume when underwriters and option holders compete in an equilibrium setting goes beyond the scope of this work.

\subsection{Applications}\label{sec:discussion-apps}
Thus far, we have considered the construction and valuation of AmPOs. Within this section, we will briefly consider use cases for these novel perpetual options. In doing so, we will consider applications both within traditional financial markets and decentralized, i.e., blockchain-based, markets. In line with the discussion in Section~\ref{sec:discussion-amort} above, we restrict our discussion herein to the constant amortization rate setting of Section~\ref{sec:bs-const}.

\subsubsection{Traditional Finance}\label{sec:discussion-tradfi}
Consider that installment options have been traded in over-the-counter markets already~\cite{kimura2009american}. However, as noted in Section~\ref{sec:ampo}, such options are non-fungible because of the specific requirements of installment payments and lapsing logic. In contrast, AmPOs provide an option type that can be transacted on a central limit order book. This exchange-based tradability permits more robust price-discovery and the possibility of more frequent transactions (due to, e.g., the ability to close long and short positions easily with market operations). 

Notably, due to the deterministic cost-of-carry of AmPOs, large institutional investors (e.g., pension plans) can hold a large perpetual protective position while only needing to pay $q V_t dt$ (or a discretized variant) per unit of notional over time to maintain the desired notional exposure. This is in contrast to dated options which require rolls around the expiry date. Such rolling operations become high-risk events due to the timing involved as, e.g., the Greeks of at-the-money options can explode near expiry.

\subsubsection{Decentralized Finance}\label{sec:discussion-defi}
While exchange-tradability is beneficial in traditional markets, it is essential for use in decentralized finance (DeFi) due to the nature of blockchain technologies. Though non-fungible tokens exist in DeFi, they are generally traded in illiquid markets, such as \emph{OpenSea}, because every contract would need to find a specific counterparty. In comparison, fungible tokens trade in liquid limit order books or automated market makers.

Automated market makers function by holding reserves of two or more assets against which anyone can transact~\cite{bichuch2022axioms}. These asset reserves are supplied by liquidity providers in exchange for the fees earned on every transaction. However, though the liquidity providers earn fees, their supplied assets are continuously exposed to trading by liquidity takers. Due to adverse selection, the effective position held by the liquidity providers (less the fees) is strictly worse than if replicated in an external market; this economic difference has been termed the loss-versus-rebalancing~\cite{milionis2022automated}. 
Notably,~\cite{singh2025modeling} proves that CI options can hedge this exposure; because AmPOs are a fungible variant of CI options, the same construction directly extends to AmPOs.

\bibliographystyle{apalike}
\bibliography{bibtex}

\end{document}